%% file: template.tex
\documentclass[10pt, conference, letterpaper]{IEEEtran}
\IEEEoverridecommandlockouts
\usepackage{cite}
\usepackage{amsmath,amssymb,amsfonts}
\usepackage[ruled,linesnumbered]{algorithm2e}
\usepackage{graphicx}
\usepackage{textcomp}
\usepackage{xcolor}
\usepackage{xspace}
\usepackage{nicefrac}
\usepackage{amsthm}
\usepackage{hyperref}
\usepackage{tabularx}
\usepackage{hhline}
\usepackage{glossaries}
\usepackage{subfig}
\usepackage{tabularx}
\usepackage{hhline}
\newtheorem{theorem}{Theorem}

\def\BibTeX{{\rm B\kern-.05em{\sc i\kern-.025em b}\kern-.08em
    T\kern-.1667em\lower.7ex\hbox{E}\kern-.125emX}}
    
\newcommand{\ie}{i.\,e.\xspace}
\newcommand{\eg}{e.\,g.\xspace}

\newcommand{\impact}{\mu}
\newcommand{\privacy}{\phi}
\newcommand{\Privacies}{\Phi}
\newcommand{\dependency}{\Lambda}
\newcommand{\error}{\Delta}

\newcommand{\ignore}[1]{}

\newcommand{\red}[1]{{\color[rgb]{1,0,0} #1}}
\newcommand{\blue}[1]{{\color[rgb]{0,0,1} #1}}

\begin{document}

\input{properties/glossary.tex}

\title{\input{properties/title.tex} \\
\input{properties/funding.tex}}

\author{\IEEEauthorblockN{Tobias Meuser\IEEEauthorrefmark{1}\IEEEauthorrefmark{3}, Oluwasegun Taiwo Ojo\IEEEauthorrefmark{3}, Daniel Bischoff\IEEEauthorrefmark{1},\\ Antonio Fern\'{a}ndez Anta\IEEEauthorrefmark{3}, Ioannis Stavrakakis\IEEEauthorrefmark{2}, Ralf Steinmetz\IEEEauthorrefmark{1}}
\IEEEauthorblockA{\IEEEauthorrefmark{1}%
Multimedia Communications Lab (KOM), Technische Universit{\"a}t Darmstadt, Darmstadt, Germany  \\
Email: \{tobias.meuser, daniel.bischoff, ralf.steinmetz\}@KOM.tu-darmstadt.de}
\IEEEauthorblockA{\IEEEauthorrefmark{2}%
National and Kapodistrian University of Athens, Greece  \\
Email: ioannis@di.uoa.gr}
\IEEEauthorblockA{\IEEEauthorrefmark{3}%
IMDEA Networks Institute, Madrid, Spain  \\
Email: \{oluwasegun.ojo@imdea.org, antonio.fernandez@imdea\}.org}}

%\author{\IEEEauthorblockN{Tobias Meuser, Oluwasegun Taiwo Ojo, Daniel Bischoff, Antonio Fern{'a}ndez Anta, Ioannis Stavrakakis, Ralf Steinmetz}
%\IEEEauthorblockA{\textit{Multimedia Communication Lab} \\
%\textit{Technische Universit\"at Darmstadt}\\
%Darmstadt, Germany \\
%\{tobias.meuser,ralf.steinmetz\}@kom.tu-darmstadt.de}
%}

\maketitle

\begin{abstract}
\input{chapter/0-abstract.tex}
\end{abstract}

\begin{IEEEkeywords}
\input{properties/keywords.tex}
\end{IEEEkeywords}
\begin{NoHyper}

\section{Introduction}
\input{chapter/1-introduction.tex}

\section{Related Work}\label{sec:related_work}
\input{chapter/2-related_work.tex}

\input{chapter/3-own_contribution.tex}

%\red{
%We need:
%\begin{itemize}
%    \item Nash Equilibrium
%    \item Optimality
%    \item Influence of deviation from the expected environment of a %neighboring vehicle
%    \item Costs of Stability, Anarchy
%    \item Quantify gain
%    \item How to share the gains between vehicle (or what could be %given between vehicles to compensate for higher bandwidth usage, %Shapley-Value)
%\end{itemize}
%}

\section{Evaluation}\label{sec:evaluation}
\input{chapter/4-evaluation.tex}

\section{Conclusion}\label{sec:conclusion}
\input{chapter/5-conclusion.tex}

% \section{Appendix}
% \input{chapter/7-appendix.tex}

\bibliographystyle{IEEEtran}
\bibliography{template}
\end{NoHyper}
\end{document}

%% file: properties/glossary.tex
\newacronym{ack}{ACK}{Acknowledgement}
\newacronym{aoi}{AoI}{Area of Interest}
\newacronym{appc}{APPC}{Adaptive Push and Pull Algorithm for Clusters}
\newacronym{bt}{BT}{Bluetooth}
\newacronym{ble}{BLE}{Bluetooth Low Energy}
\newacronym{c2c}{C2C}{Car-to-Car}
\newacronym{cbdc}{CBDC}{Cluster Based Data Consistency}
\newacronym{cbrm}{CBRM}{Cluster Based Replica Management}
\newacronym{cs}{CS}{Content Store}
\newacronym{csv}{CSV}{Comma-separated Values}
\newacronym{d2d}{D2D}{Device-to-Device}
\newacronym{dis}{DIS}{Data Integration System}
\newacronym{dos}{DoS}{Denial-of-Service}
\newacronym{dsrc}{DSRC}{Dedicated Short-Range Communication}
\newacronym{dtn}{DTN}{Delay-Tollerant Networking}
\newacronym{fcc}{FCC}{Federal Communications Commission}
\newacronym{gcc}{GCC}{Global Cluster Cooperation}
\newacronym{gps}{GPS}{Global Positioning System}
\newacronym{gui}{GUI}{Graphical User Interface}
\newacronym{icn}{ICN}{Information-Centric Networking}
\newacronym{id}{ID}{identifier}
\newacronym{ieee}{IEEE}{Institute of Electrical and Electronics Engineers}
\newacronym{iot}{IoT}{Internet of Things}
\newacronym{ip}{IP}{Internet Protocol}
\newacronym{lte}{LTE}{Long Term Evolution}
\newacronym{lte-a}{LTE-A}{Long Term Evolution-Advanced}
\newacronym{lru}{LRU}{Least Recently Used}
\newacronym{madu}{MADU}{Mobility-Aware Data Update}
\newacronym{manet}{MANET}{Mobile Ad-hoc Network}
\newacronym{ndn}{NDN}{Named Data Networking}
\newacronym{obu}{OBU}{On-Board Unit}
\newacronym{p}{P}{Process}
\newacronym{p2p}{P2P}{Peer-to-Peer}
\newacronym{poa}{POA}{Place of Action}
\newacronym{poi}{POI}{Point of Interest}
\newacronym{qoe}{QoE}{Quality of Experience}
\newacronym{qos}{QoS}{Quality of Service}
\newacronym{rsu}{RSU}{Road Side Unit}
\newacronym{rtt}{RTT}{Round-Trip Time}
\newacronym{rttl}{RTTL}{Remaining Time to Live}
\newacronym{slaw}{SLAW}{Self-similar Least Action Walk}
\newacronym{tj}{TJ}{Traffic Jam}
\newacronym{ttl}{TTL}{Time to Live}
\newacronym[\glsshortpluralkey=VANETs,\glslongpluralkey=Vehicular Ad-hoc Networks]{vanet}{VANET}{Vehicular Ad-hoc Network}
\newacronym{wave}{WAVE}{Wireless Access in Vehicular Environments}
\newacronym{wifi}{WiFi}{Wireless Fidelity}
\newacronym{wlan}{WLAN}{Wireless Local Area Network}

\newacronym{most-common-information}{MJ}{Majority Voting}
\newacronym{newest-information}{NI}{Newest Information}

\newacronym[plural={ADAS}]{adas}{ADAS}{Advanced Driver Assistance System}

\newacronym{hmm}{HMM}{Hidden Markov Model}

\newacronym{qoi}{QoI}{Quality of Information}

\newacronym{hs}{HS}{Hidden State}
\newacronym{os}{OS}{Observable State}

\newacronym{sumo}{SUMO}{Simulation of Urban Mobility}
\newacronym{traci}{TraCI}{Traffic Control Interface}

\newacronym{wsn}{WSN}{Wireless Sensor Network}
\newacronym{gnss}{GNSS}{Global Navigation Satellite System}
\newacronym{lidar}{LIDAR}{Light Detection and Ranging}

\newacronym{pubsub}{Pub/Sub}{Publish/Subscribe}

\newacronym{cam}{CAM}{Cooperative Awareness Message}
\newacronym{den}{DEN}{Decentralized Environment Notification}
\newacronym{denm}{DEN}{Decentralized Environment Notification Message}
\newacronym{ivc}{IVC}{Inter-Vehicle Communication}
\newacronym{fcd}{FCD}{Floating Car Data}
\newacronym{v2v}{V2V}{Vehicle to Vehicle}
\newacronym{ch}{CH}{Cluster Head}
\newacronym{mbms}{MBMS}{Multimedia Broadcast/Multicast Service}

%% file: properties/title.tex
Hide Me: Enabling Location Privacy in Heterogeneous Vehicular Networks

%% file: properties/funding.tex
% \\ \thanks{Identify applicable funding agency here. If none, delete this.}

%% file: chapter/0-abstract.tex
%Today's vehicles use location-based services to increase their awareness beyond their local perception. For this purpose, the vehicles share their location with a central entity to receive relevant data. This sharing of location data compromises the privacy of the occupants. Our approach to privacy protection is the obfuscation of location data, \ie, adding artificial noise to the provided location. While this protects the privacy of the occupants, it decreases the performance of the location-based service, as the server-side filtering becomes less accurate.
%
%In this work, we present our innovative approach to compensate for this performance decrease through cooperation between vehicles. Compared to previous works, we do not rely on the clustering of vehicles, which is prone to frequent disconnects in urban areas. Instead, we utilize a game-theoretic approach, in which every vehicle follows the same strategy. Hence, the disappearance of a single vehicle impacts the performance only marginally, and can be compensated easily. In our analytical and numerical analysis, we show that the performance increase is significant compared to non-cooperative approaches and the performance decrease induced through privacy can be bounded In the evaluation, we show that our approach significantly outperforms clustering approaches through the higher robustness to topology changes.
To support location-based services, vehicles must share their location with a server to receive relevant data, compromising their (location) privacy. To alleviate this privacy compromise, the vehicle's location can be obfuscated by adding artificial noise. Under limited available bandwidth, and since the area including the vehicle's location increases with the noise, the server will provide fewer data relevant to the vehicle's true location, reducing the effectiveness of a location-based service. To alleviate this problem, we propose that data relevant to a vehicle is also provided through direct, ad hoc communication by neighboring vehicles. Through such Vehicle-to-Vehicle (V2V) cooperation, the impact of location obfuscation is mitigated.  Since vehicles subscribe to data of (location-dependent) impact values, neighboring vehicles will subscribe to largely overlapping sets of data, reducing the benefit of V2V cooperation. To increase such benefit, we develop and study a non-cooperative game determining the data that a vehicle should subscribe to, aiming at maximizing its utilization while considering the participating (neighboring) vehicles. Our analysis and results show that the proposed V2V cooperation and derived strategy lead to significant performance increase compared to non-cooperative approaches and largely alleviates the impact of privacy on location-based services.

%% file: properties/keywords.tex
Floating Car Data, location-based services, location privacy, V2V communication.

%% file: chapter/1-introduction.tex
The vehicles of the future will be required to have increased awareness about their environment to assist drivers or support autonomous driving. This awareness has typically been provided by different sensors on board the vehicles, measuring vital data about the environment of the vehicle.
%This data is in turn used in the planning of routes and trajectories.
The context data provided by these sensors is limited to the vehicle's immediate environment due to the sensors' inherent physical limitations, such as their range. Nevertheless, information pertaining to an environment not in the immediate vicinity of a vehicle is also important for traffic safety, routing and navigation, as such vehicle can consider upcoming events in their planning. To make such information available to far away vehicles, first, vehicles receiving this information through local detection should communicate it to a remote server using an appropriate communications infrastructure such as the cellular network. Then, the vehicles desiring to receive such information should indicate so to the server, and receive it via a similar infrastructure. By exchanging their local perception of the environment via a cellular infrastructure, vehicles can supplement their local perception with distant data provided by other vehicles. 

In order for vehicles to get \gls{fcd} that are relevant to their context, they have to continuously share their location with a centralized back-end server in the network that is assumed to be a trusted entity. The server filters out the relevant \gls{fcd} for vehicles based on their location and distributes them accordingly. This continuous context and location exchange with a server is an invasion to the privacy of the vehicles. Consequently, privacy-sensitive users either have to accept this privacy invasion or turn off the option of receiving \gls{fcd}. Clearly, users that decide to turn off receipt of \gls{fcd} cannot benefit from location-based services and other services enabled by vehicular networks. It is therefore desirable to have a mechanism that allows for effective transfer of \gls{fcd} to and from the network while preserving the privacy of the users involved, like adding noise to the location (obfuscation). 

%To overcome this challenge while still enabling location-based services, we propose a new approach to exchanging and offloading \gls{fcd} among vehicles while taking into account their privacy concerns. 

%\blue{To support location-based services, vehicles share their location with a server in order to receive relevant data, compromising their (location) privacy. A typical approach to alleviate this privacy compromise, is to obfuscate the vehicle’s location by adding artificial noise. As a result, the area that includes the location of a privacy-concerned vehicle increases and, in principle, the server will need to provide all the data relevant to that larger area. Since the available bandwidth is typically limited, the latter leads to actually providing fewer data relevant to the vehicle's true location, as explained below.  In this paper we assume that vehicles can have different levels of privacy protection, impacting on the level of location data obfuscation accordingly.  }

Since not all data are equally important to a vehicle, they are classified based on some impact level. Higher impact data are prioritized by the server since the available bandwidth is limited. In practice, vehicles subscribe to data of some (location-dependent) impact level, and the server provides to the vehicle all available data with matching impact. Privacy-concerned vehicles that do not report their exact location will end up receiving a smaller portion of data of a given impact due to the need to accommodate such data relevant to a larger geographic area. As a result, location-based services would be less effectively provided to privacy-concerned vehicles. 

To alleviate this problem and increase the amount of  location-relevant data provided to the vehicles, we propose that data relevant to a vehicle's location is also provided by neighboring vehicles through direct, ad hoc communication. That is, consider Vehicle-to-Vehicle (V2V) cooperation for exchanging local relevant data. Through this, the impact of location obfuscation (\ie increased area containing the vehicle's true location) is mitigated to some extent.  Such V2V cooperation based on vehicle clusters has also been considered in \cite{MBRS19}. As it will be discussed later and shown in the results, besides the higher complexity of a cluster-based approach, it also suffers greatly from connectivity problems reducing greatly its performance. In the current work, besides incorporating privacy considerations, the V2V communication is not cluster-based but ad hoc through direct V2V exchanges. 

Notice that neighboring vehicles without coordination are expected to subscribe to largely overlapping sets of data, reducing the potential benefit of V2V cooperation. To increase such benefit potential, we develop and study a non-cooperative game determining the data that a vehicle should subscribe to, so that the aforementioned overlap is reduced. The design goal is to maximize a properly defined utilization function as shaped by the participating (neighboring) vehicles as well. Our analysis and results show that the proposed V2V cooperation scheme and derived strategy lead to significant performance increase compared to non-cooperative approaches and largely alleviates the impact of privacy on location-based services.

%This work a direct extension to our work \cite{MBRS19}, which is a method to offload and disseminate context information supplied by other vehicles on the network, in the form of \gls{fcd}. 

%To achieve this, we propose a communication mechanism with adjustable privacy in which the users can set the level of privacy for the communication while still receiving the most crucial information. This is done by concealing the real location of a privacy-sensitive vehicle by obfuscating and increasing the occupancy area of the vehicle. Under the assumption that a vehicle has limited bandwidth for the exchange of \gls{fcd}, the imprecision of context submitted by the vehicle leads to less share of relevant \gls{fcd} being served by the back end. We mitigate against this through an implicit cooperation among vehicles and we utilize a game theoretic approach to model the strategies to maximize the expected impact of messages received by the vehicles.

In summary, the contributions of this work are the following: First, we introduce privacy considerations in the management of \gls{fcd} and point to its impact on location-based services; given a fixed bandwidth availability, some data are not forwarded to a vehicle due to privacy considerations and the implemented location obfuscation. Second, and in order to alleviate the latter problem, we propose that vehicles cooperate and forward relevant data to their neighboring vehicles, enhancing in principle the data received by a vehicle besides those directly from the remote server; an ad-hoc, direct V2V cooperation paradigm is employed instead of a cluster-based one, also showing the high performance deterioration of the latter in a real vehicular networking environment. Third, a major contribution is the development and study of a non-cooperative game determining the strategies (in terms of probabilities that a vehicle is forwarded by the server data of a given impact index) that vehicles should follow, so that a properly defined utility is maximized; this is shown to lead to a diversification of the data received directly from the server by neighboring vehicles and increases the effectiveness of V2V cooperation. Finally, the aforementioned contributions are supported through numerical and simulation results.

The rest of the article follows thus: In \autoref{sec:related_work}, we give a brief overview of some related work in the literature. In \autoref{sec:system_overview}, we provide an overview of the system scenario while in \autoref{sec:location_privacy}, we describe the influence of location privacy on the network. In \autoref{sec:game}, we describe our proposed game theoretic approach for privacy sensitive communication and we provide the necessary adaptations and assumptions necessary for the method in \autoref{sec:analysis}. In \autoref{sec:evaluation}, we evaluate the performance of our method and we conclude this paper in \autoref{sec:conclusion} with some discussions about our findings.

%% file: chapter/2-related_work.tex
Since vehicles have to continuously update their location to a centralized back-end server to receive \gls{fcd} that is relevant to their context. \ignore{, there is an implicit assumption that the server is being run by a trusted entity. However, in the case that an adversary gets control of such server, the adversary has access to all the information about a user, like their current location, favourite locations to visit, etc.} This is a breach of privacy for the users as their personal data can be shared with third parties. Several techniques have been introduced in the literature to protect users' privacy in vehicular networks\ignore{ and in location-based services}. Some of the common techniques include the use of pseudonyms \cite{Golle2004_1,Dotzer2005_1,7289480}, obfuscation \cite{xpan_1,7962836}, and the use of group communications \cite{Wasef2010_1,Sampigecaravan,7414507}.

The use of pseudonyms involves users taking on another identity (pseudonyms) to dissociate their actual identity from their data \cite{petit2014pseudonym}. The use of a single pseudonym is less effective and it is often required for users to change pseudonyms through their journey to sustain their level of privacy \cite{Gerlach_2007}. Such pseudonym changes are usually done in mix zones where drivers have to switch pseudonyms \cite{Palanisamy_2011}. These mix zones can be fixed \cite{Freudiger_109437} or specified dynamically \cite{BYing_1}. However, the use of pseudonyms have been shown to be less effective against a global eavesdropper\cite{wiedersheim2010privacy}, and especially in environments with low car density like on highways.  Furthermore, the use of pseudonyms usually focuses  on eavesdroppers monitoring V2V communications and involves having to deal with a trusted (or semi-trusted) server which coordinates the assignments of pseudonyms \cite{Sampigecaravan}. This still involves trusting a central server which is a risk in the case that an adversary gets hold of such server. Our work focuses on the privacy of users in their communications with the central server.  

Likewise, obfuscation has been used extensively in privacy protection in vehicular networks and location-based services. Obfuscation involves users providing one (i) an inaccurate location, (ii) an imprecise region including their real location, or (iii) a vague description of their location \cite{duckham2006location}. To quantify the effectiveness of obfuscation, metrics like \textit{k-anonimity}, which means a user's location is indistinguishable from $k-1$ other users, have been introduced \cite{Gruteser2003AUL,BenNiu}. The imprecision added into the location of the users usually leads to users getting less relevant data and, thus, a decrease in efficiency. Our method mitigates against this decrease in performance by implicitly cooperating with other vehicles to get relevant updates through V2V communication. 

For this purpose, game-theory has been applied to model aspects of privacy, especially in mobile networks and location-based services \cite{XLiu_Game,Freudiger_Game} and in security and privacy assessment of vehicular networks \cite{du2014attack}. Distinct from previous studies, our work focuses on privacy of users in their communications with the server considering the impact of the messages to the user. We adopt an obfuscation technique by reporting a region instead of their exact location and mitigate against the resulting reduction in performance by implicitly cooperating the vehicles through a game-theoretic approach, which maximizes the relevant data received by the vehicles. 

%% file: chapter/3-own_contribution.tex
\section{System Model}\label{sec:system_overview}
\input{snippets/system/general.tex}
\section{Influence of Location Obfuscation}\label{sec:location_privacy}
\input{snippets/modeling/bandwidth.tex}
\input{snippets/modeling/bandwidth_proof.tex}
\section{Game-Theoretic Model for Privacy-Sensitive Communication}\label{sec:game}
\input{snippets/game/description.tex}
\subsection{Game-Theoretic Solution}
\input{snippets/game/optimal_strategy.tex}
\subsection{Deriving the Utility-Optimal Strategy}\label{subsec:optimal}
\input{snippets/game/solution.tex}

\section{Analysis and Required Adaptation}\label{sec:analysis}
\input{snippets/analysis/intro.tex}

\subsection{Properties of the Developed Solution}
\input{snippets/analysis/properties.tex}

% \subsection{Formal Analysis and Bounds}
% \input{snippets/analysis/formal.tex}

\subsection{Numerical Analysis of Influence Factors}\label{subsec:numerical}
\input{snippets/analysis/numerical.tex}

%% file: snippets/system/general.tex
In the following, we provide an overview of the considered system model.
We assume a context-aware vehicular network, in which a central server transmits context-sensitive messages to interested vehicles.
In this network, time is assumed to be slotted (typically $1s$).
Every vehicle has a limited average bandwidth $A$ (in bits per time slot) to receive these messages via the cellular network.
This assigned bandwidth is generally low compared to the maximum (physically) available bandwidth, such that vehicles may exceed this bandwidth temporarily (as long as the average consumed bandwidth matches the predefined value).
% The interest of vehicles is determined by matching their context (location) with the meta-information of the transmitted messages.
A message contains \gls{fcd} as payload, as well as additional meta-information such as the measurement location, measurement time, and type of \gls{fcd}.
In this work, we assume that \gls{fcd} refer to or carry road-related information (e.g., accidents, traffic jams, traffic flow information, etc.) that can be useful for improving the driving behavior of the vehicles in proximity. 
Let $a(m)$ (in bits) denote the size of such message $m$, $r(m)$ the radius of its dissemination area, and $\impact(m)$ its impact, which all depend on the type of contained \gls{fcd}; \eg, an accident has generally higher impact than traffic flow information.
We use this notation as our bandwidth is limited, and thus, the impact per utilized bandwidth is pivotal for our approach.
Based on the message impact $\impact(m)$ and the dissemination area $r(m)$, we divide messages in $n_\impact$ impact levels; an impact level $i \in \{1, \hdots, n_\impact\}$ contains messages if $r(m)=r_i$ and $\impact(m) \in [\impact_i, \impact_{i+1})$, where $r_i$ is radius of the dissemination area and $\impact_i$ the impact per bit assigned to impact level $i$.
For every impact level, there is an expected impact $\overline{\impact}_i$, which refers to the average impact per bit of a message in this impact level.
For $i = n_\impact$, we set the upper bound $\overline{\impact}_{i+1} = \infty$.

%We assume that we can determine the size $a(m)$ (in bits) of this message and an impact $\impact(m)$ of a message $m$, which depend on the type of contained \gls{fcd}, \eg, an accident has generally more impact than traffic flow information.
%Based on the message impact $\impact(m)$, we divide messages in $n_\impact$ impact levels; an impact level $i \in \{1, \hdots, n_\impact\}$ contains messages if $\nicefrac{\impact(m)}{a(m)} \in [\impact_i, \impact_{i+1})$ where $\impact_i$ is the impact value per bit assigned to impact level $i$.
% every impact level, there is an expected impact $\overline{\impact}_i$, which refers to the average impact of a message in this impact level.
% For $i = n_\impact$, we set the upper bound $\overline{\impact}_{i+1} = \infty$.

A vehicle can control the reception of messages from the server by expressing interest in certain \textit{impact} levels and by providing a \textit{representation of its location}. More specifically, a vehicle wants to receive a message $m$ if (i) it has expressed interest in the corresponding impact level $i$ of the message (ii) the vehicle's location is at most at distance $r_i$ from the source of the \gls{fcd}.
Let $a_i$ denote the traffic load of messages of impact level $i$ (in bits per time slot) expected for the vehicle if the provided location is accurate.
A vehicle is either interested in an impact level or not, \ie, receives either all or no messages of this impact level.
This interest can be changed dynamically at the beginning of every time slot.
Depending on the assumed privacy-sensitivity (referred to as privacy-level) $\privacy \in \Privacies$ of a vehicle $v$, the aforementioned \textit{representation of the location} may be accurate or may be imprecise.
We implement this imprecision by providing only an (circular) area in which the vehicle is certainly located (uniformly distributed), without actually revealing the exact location to the server.
The privacy level $\privacy$ chosen by the respective vehicle determines the radius $r_\privacy$ of this area.
That imprecise representation of the location increases the load of received messages due to the less accurate server-side filtering.
To capture the additional bandwidth consumption, let $a_{\privacy,i} > a_i$ denote the load of messages (in bits) expected in a certain by a vehicle impact level $i$ for a vehicle at a privacy level $\privacy$.

The central server uses the announced interest of the vehicles to actively push new messages (\ie, messages containing yet unknown \gls{fcd}) via the cellular network to them. Since the available bandwidth is assumed to be limited, a vehicle aims to maximize the total impact of the received messages, which is achieved by dropping low-impact messages if the bandwidth is insufficient.
To maximize that total impact of received messages, vehicles may cooperate to share bandwidth for the reception of messages;  \ie, vehicles can locally broadcast messages, received via the cellular network, without additional costs to notify vehicles in their proximity.
Thus, not every vehicle needs to receive all messages of its interest via the limited cellular bandwidth, as these messages might be provided by its neighbors.

%% file: snippets/modeling/bandwidth.tex
% Bandwidth-Modeling
% Requires: Privacy-Level, Bandwidth-Assumption, Impact-Definition, Context-Sharing, Server-Communication, Context-Filtering, Context-Sensitivity
In the following, we provide an insight on the influence of privacy in our model.
Each privacy level $\privacy > 1$ adds a certain level of imprecision to the provided location, while $\privacy=1$ refers to no privacy-sensitivity.
The privacy-sensitivity and, thus, location imprecision increases with $\privacy$ and reduces the accuracy of the context-based message filtering at the server-side.
Thus, a vehicle receives messages not relevant for its current context, while its share of relevant messages is reduced.
This influences the number of received messages $n_{\privacy,i}$ and their expected impact per bit $\overline{\impact}_{\privacy,i}$ for a privacy state $\privacy$ and an impact-level $i$. 
The number of messages received typically increases with increasing privacy level, while the expected impact per bit of a message decreases. 
We reflect this change for every impact level $i$ by the \textit{adaptation factor} $\rho_{\privacy, i}$  as shown in \autoref{eq:bandwidth_privacy} and
\autoref{eq:impact_privacy}.
\begin{equation}\label{eq:bandwidth_privacy}
    a_{\privacy, i} = a_i \cdot \rho_{\privacy,i}
\end{equation}
\begin{equation}\label{eq:impact_privacy}
    \overline{\impact}_{\privacy, i} = \dfrac{\overline{\impact}_i}{\rho_{\privacy,i}}
\end{equation}
$\rho_{\privacy,i}$ depends on the context-sensitivity of the distributed messages for a vehicle of privacy level $\privacy$ receiving messages with impact level $i$.
For non-context-sensitive messages, $\rho_{\privacy,i} = 1,\forall \privacy \in \Privacies$.
For context-sensitive of messages, \ie, a messages with a specific distribution-area with radius $r_i$, $\rho_{\privacy,i} \geq 1, \forall \privacy \in \Privacies$.
These statements are proven in \autoref{proof:bandwidth_decrease}.

%% file: snippets/modeling/bandwidth_proof.tex
% Bandwidth-Proof
% Requires: Privacy-Level, Bandwidth-Assumption, Context-Filtering, Bandwidth-Modeling
\begin{theorem}\label{proof:bandwidth_decrease}
The \textit{adaptation factor} for a network with uniformly distributed messages is $\rho_{\privacy,i}=\left(\nicefrac{r_\privacy}{r_i} + 1\right)^2$ for a circular geocast-area and a circular location-imprecision, where $r_i$ is the radius of the geocast-area of the message of impact level $i$ and  $r_\privacy$ is the radius of the location-imprecision area of privacy-level $\privacy$.
\end{theorem}
\begin{proof}
Without location privacy, the vehicle receives all messages with a maximum distance of $r_i$ to its current location. Thus, area of interest for the vehicle is $\pi \cdot r_i^2$.
If the vehicle reduces the precision of its location by hiding inside an area of radius $r_\privacy$, the server will need to transmit all messages within a distance of $r_\privacy + r_i$ from the center of the area to ensure that the vehicle receives all relevant messages.
The size of this area is $\pi \cdot (r_\privacy + r_i)^2$.
This leads to $\rho_{\privacy,i} = \left(\nicefrac{r_\privacy}{r_i} + 1\right)^2$.
\end{proof}

%% file: snippets/game/description.tex
% Game-Description
% Requires: Privacy-Levels, Impact Subscription, Bandwidth-Assumption, Bandwidth-Modeling, Similar-Environment-Assumption???
To enhance the performance of our impact-aware vehicular network, \ignore{(THE CLUSTER BASED ISSUES IS NOT THE REASON WE USE GAMES. WE USE COOPERATION UNDER NON-COOPERATIVE GAMES IN ORDER TO ENHANCE THE IMPACT OF RECEIVED MESSAGES VIA CELLULAR ONLY. WE DO NOT USE GAMES BECAUSE CLUSTER-BASED IS AN ISSUE. SEE OUTLINE OF CONTRIBUTIONS IN INTRO AND MAKE SURE WE ARE CONSISTENT WITH OUR STORY EVERYWHERE)} we employ a game-theoretic model with the aim to maximize the sum of impact of the received messages. 
Our innovative approach relies only on the number $n_\privacy$ of vehicles of each privacy-level $\privacy$ in proximity to find a mixed Nash-optimal solution for our developed game-theoretic model, \ie, vehicles receive messages with a certain probability.
In our game, each actor (vehicle) aims to find the strategy (receive messages in a certain impact-range via the cellular network) that maximizes its utility (sum of impact values of all received messages, directly via cellular or from the neighbors) while sticking to cellular bandwidth constraints.
This game is played periodically in every time slot to adjust the vehicles behavior to environmental changes, \ie changes in the number of neighbors in proximity and changes in number of messages.
Notice that vehicles are assumed to cooperate; thus, a vehicle might additionally receive messages directly by vehicles in proximity.
The intuition behind this game model is that high-impact messages are generally prioritized, as their bandwidth usage is more efficient compared to low-impact messages.
Thus, vehicles may rely on their neighbors to provide some high-impact messages to them, as a number of neighbors aims to receive these high-impact messages.
These vehicles can then use a part of their available cellular bandwidth to receive low-impact messages and share these with their neighbors. The idea is similar to cooperative caching. Instead of storing all high-demand message at every local cache, some nodes fetch low-demand messages instead and satisfy the request of high-demand messages from nearby cooperative caches~\cite{1717403}.\ignore{INSTEAD OF HAVING ALL HIGH DEMAND MESSAGES BROUGHT TO EACH OF THE LOCAL CACHE OF COOPERATING NODES, BRING SOME LOW DEMAND INSTEAD AND RETRIEVE ANY HIGH DEMAND NOT LOCALLY AVAILABLE FROM NEAR-BY COOPERATING CACHE; THE HIT RATIO IS IMPROVED FOR ALL CACHES THIS WAY. REF: N. Laoutaris, O. Telelis, V. Zissimopoulos, I. Stavrakakis, "Distributed Selfish Replication," IEEE Transactions on Parallel and Distributed Systems, vol. 17, no. 12, pp. 1401-1413, Dec., 2006.}

The vehicles are the only \textit{actors} in this game; the server is not directly involved, but only determines the set of receivers of messages based on the strategies chosen by the vehicles.
For this purpose, the vehicles share their strategy in the form of subscriptions with the server.
The \textit{strategy} is represented as a vector $\vec{p_\privacy}$ with $n_\impact$ probability entries $p_{\privacy, i}$ with $i \in \{1, \hdots, n_\impact\}$ and depends on the chosen privacy level of the vehicle $\privacy_e$.
Each entry $p_{\privacy,i}$ refers to the probability of the tagged vehicles to receive messages of the corresponding impact level.
Additionally, $0 \leq p_{\privacy,i} \leq 1,\forall p_{\privacy,i} \in \vec{p_\privacy}$.
For the assignment of messages to an impact level, we use the expected impact $\overline{\impact}_i$, \ie, the average impact of a message of that type given it is relevant for the vehicle.
$\overline{\impact}_i$ does not depend on the privacy level $\privacy$.
The privacy-dependent message impact $\overline{\impact}_{\privacy,i}$ is only used for the calculation of the utility for a vehicle.
In the calculation, $\vec{p_\privacy}$ needs to be chosen such that \autoref{eq:bandwidth} holds, with $a_{\privacy,i}$ being the expected number of received messages of impact level $i$ and privacy level $\privacy$ according to \autoref{eq:bandwidth_privacy}, and $A$ being the usable bandwidth.

\begin{equation}\label{eq:bandwidth}
    \sum_{i=1}^{n_\impact} a_{\privacy,i} \cdot p_{\privacy,i} \leq A
\end{equation}

\ignore{\blue{Notice that this differs from previous work like \cite{MBRS19}, in which the vehicle intended to receive all messages $\{m | \impact_i \leq \impact(m)\}$.
The advantage of our new model is the additional flexibility provided by removing the overlap between strategies, which improves the coordination between vehicles.}}

Notice that this differs from previous work like \cite{MBRS19}, in which the vehicle intended to receive all messages $\{m | \impact_i \leq \impact(m)\}$.
The advantage of our new model is the additional flexibility provided by removing some of the message overlap among neighboring vehicles, which improves the total impact of received messages (via cellular and direct neighbor forwarding) by each vehicles.

The \textit{utility} each vehicle aims to maximize captures the value of the received messages, \ie, refers to the sum of impact.
The utility is defined according to \autoref{eq:utility} based on the sent messages $M_{snt}$, the received messages $M_{rcv}$, and the impact $\impact(m)$ of a message $m$.
$\mathbb{I}_{\{m \in M_{rcv}\}}$ is an indicator function indicating whether a message $m$ has been received by the vehicle.

\begin{equation}\label{eq:utility}
    u = \sum_{m \in M_{snt}} \impact(m) \cdot a(m) \cdot \mathbb{I}_{\{m \in M_{rcv}\}}
\end{equation}

As the probability of a vehicle receiving a message depends on $\vec{p_\privacy}$, we derive the expectation of the utility based on \autoref{eq:utility}.
For this purpose, we assume that the environment of each vehicle is rather similar, so that the strategies of two vehicles with the same privacy level are similar.
Thus, the strategy of every privacy level can be calculated by every vehicle in proximity, which is the basis of our offloading approach.
Thus, we only use the strategies $\vec{p_\privacy}$ along with the number $n_\privacy$ of vehicles for each privacy level $\privacy$ to calculate the probability of receiving a message either via the cellular network or from one of the neighbors.
The probability $p(\impact_i)$ to receive a message via any interface (Cellular or Wifi) with at impact level $\impact_i$ can be calculated as shown in \autoref{eq:receive_probability}.
This formula assumes that there is no loss in the network, \ie, every transmitted messages is received by the intended receiver.

\begin{equation}\label{eq:receive_probability}
    p(\impact_i) = 1 - \prod_{\privacy \in \Privacies} (1 - p_{\privacy,i})^{n_\privacy}
\end{equation}

We use the probability $p(\impact_i)$ to receive a message to derive the expected utility $\overline{u}(\privacy_e, \vec{p_1}, \hdots, \vec{p_{\left|\Privacies\right|}})$.
This estimates the set of received messages $M_{rcv}$ using the expected amount of sent messages $a_i$ and the probability $p(\impact_i)$ to receive each message.
The resulting expected utility for the tagged vehicle is shown in \autoref{eq:expected_utility}.

\begin{equation}\label{eq:expected_utility}
    \overline{u}(\privacy_e, \vec{p_1}, \hdots, \vec{p_{\left|\Privacies\right|}}) = \sum_{i = 1}^{n_\impact} \overline{\impact}_{\privacy_e,i} \cdot a_{\privacy_e,i} \cdot \left[1 - \prod_{\privacy \in \Privacies} (1 - p_{\privacy,i})^{n_\privacy}\right]
\end{equation}
When clear from context, we refer to $\overline{u}(\privacy_e, \vec{p_0}, \hdots, \vec{p_{\left|\Privacies\right|}})$ as $\overline{u}$ to increase readability.
In the next section, we describe the process of deriving a utility-maximizing strategy for the described game.
The advantage of determining the solution analytically is (i) the possibility to analyze and bound the effects of location privacy to the system, and (ii) the lower computational complexity compared to a non-linear solver.

%% file: snippets/game/optimal_strategy.tex
% Game-Description
% Requires: Game-Description, Bandwidth-Assumption, tagged vehicle, (Non-)Zeroness of Probabilities,
In this section, we derive the optimal strategy for a vehicle with privacy level $\privacy_{e}$ given that the privacy level and number of vehicles in each privacy level in its environment is known.
For this purpose, we calculate the partial derivatives of the expected utility $\overline{u}$ with respect to the probabilities of the tagged vehicle $p_{\privacy,i}$.
However, it is important to consider the dependency between the probabilities $p_{\privacy,i},\forall \privacy \in \Privacies$, as \autoref{eq:bandwidth} limits the possible values of $p_{\privacy,i}$.
This approach would work similarly with any other probability $p_{\privacy,i}|i \neq 1$.
We depict this dependency by expressing $p_{\privacy,1}$ depending on the other probabilities $\{p_{\privacy,i}|i > 1$\} as shown in \autoref{eq:p1}.
Thus, $p_{\privacy,1}$ depends on all other probabilities, \ie, the derivative of $p_{\privacy,1}$ with respect to any probability $p_{\privacy,i}$ is not always non-zero, which leads to our optimization problem.

\begin{equation}\label{eq:p1}
    p_{\privacy,1} \leq \dfrac{A - \sum_{i=2}^{n_\impact} a_{\privacy,i} \cdot p_{\privacy,i}}{a_{\privacy,1}}
\end{equation}

While the inequality is sufficient to guarantee the bandwidth requirements, we will assume \autoref{eq:p1} to be an equation as higher values of $p_{\privacy,1}$ cannot decrease the utility.
As there is no dependency between any pair of probabilities $p_{\privacy,i}$ and $p_{\privacy,j}$ if $i \neq j \land i \neq 1 \land j \neq 1$, the derivative of the utility with respect to $p_{\privacy,l}$ depends only on $p_{\privacy,1}$ and $p_{\privacy,l}$ for every $l > 1$ as shown in \autoref{eq:expected_utility_derivative}.
Notice that $\overline{\impact}_{\privacy_e, i} \cdot a_{\privacy_e, i} = \overline{\impact_i} \cdot a_i$ according to \autoref{eq:bandwidth_privacy} and \autoref{eq:impact_privacy}.
Additionally, we assume that $p_{\privacy_e,l} \neq 0$.
We ensure that by considering the cases with $p_{\privacy_e,l} = 0, \forall l \in \{1, \hdots, n_\impact\}$ separately as described in \autoref{subsec:optimal}.

\begin{multline}\label{eq:expected_utility_derivative}
    \dfrac{\partial \overline{u}}{\partial p_{\privacy_{e},l}} = \overline{\impact}_l a_l n_{\privacy_{e}} \cdot (1 - p_{\privacy_{e},l})^{n_{\privacy_e} - 1} \cdot P_l(\Privacies \setminus \{\privacy_e\}) +
    \\ \overline{\impact}_1 a_1 \left(\dfrac{\partial p_{\privacy_e, 1}}{\partial p_{\privacy_e, l}}\right) n_{\privacy_{e}} \cdot (1 - p_{\privacy_{e},1})^{n_{\privacy_e} - 1} \cdot P_1(\Privacies \setminus \{\privacy_e\})
\end{multline}
with
$$ P_j(\Privacies) = \prod_{\privacy \in \Privacies} (1 - p_{\privacy,j})^{n_\privacy} $$

\autoref{eq:bandwidth_privacy} displays the dependency of $p_{\privacy_e, 1}$ and $p_{\privacy_e, l}$.
Thus, the derivative of $p_{\privacy_e, 1}$ with respect to $p_{\privacy_e, l}$ can be calculated according to \autoref{eq:p1_for_pl}.

\begin{equation}\label{eq:p1_for_pl}
    \dfrac{\partial p_{\privacy_e, 1}}{\partial p_{\privacy_e, l}} = -\dfrac{a_{\privacy_e,l}}{a_{\privacy_e,1}}
\end{equation}

By setting the derivative of the utility to $0$, we determine all possibly optimal solutions.
This leads to \autoref{eq:base_equation_probability} after some minor transformations.
Notice that $a_l$ and $n_{\privacy_e,}$ are omitted as they are present on both sides of the equation.

\begin{multline}\label{eq:base_equation_probability}
    \overline{\impact}_l \cdot (1 - p_{\privacy_{e},l})^{n_{\privacy_e} - 1} \cdot P_l(\Privacies \setminus \{\privacy_e\})
    \\ = \overline{\impact}_1 \cdot \dfrac{\rho_{\privacy_e, l}}{\rho_{\privacy_e, 1}} \cdot (1 - p_{\privacy_{e},1})^{n_{\privacy_e} - 1} \cdot P_1(\Privacies \setminus \{\privacy_e\})
\end{multline}

For a given impact level $l$, we divide the set of privacy levels $\Privacies$ into $\Privacies^+(l)$, which only contains privacy levels with $p_{\privacy, l} > 0$, and $\Privacies^-(l)$, which contains privacy levels with $p_{\privacy, l} = 0$.
This is necessary, as the derivative of the expected utility with respect to $p_{\privacy, l}$ is always $0$ if $p_{\privacy, l}=0$, thus, \autoref{eq:base_equation_probability} does not hold.
However, \autoref{eq:base_equation_probability} still contains $p_{\privacy, l}, \forall \privacy \in \Privacies^+(l)$ and $p_{\privacy, 1}, \forall \privacy \in \Privacies(l)$.
We need to replace $p_{\privacy, l}, \forall \privacy \in \Privacies^+(l)\setminus \privacy_e$ to calculate $p_{\privacy_e, l}$.
We can calculate the $p_{\privacy_e, l}$ using \autoref{eq:probability_calculation} according to \autoref{proof:simplification}.

\begin{theorem}\label{proof:simplification}
For the probability $p_{\privacy_e, l}, \forall \privacy_e \in \Privacies^+(l)$ with $n_{\privacy_e}>1$, \autoref{eq:probability_calculation} holds. It depends only on $p_{\privacy_e, l}$ and previously calculated probabilities and can be used to calculate $p_{\privacy_e, l}$.
\begin{multline}\label{eq:probability_calculation}
    \impact_l \cdot \left( 1 - p_{\privacy_e, l} \right)^{n^+(l)} \\ = \impact_1 \cdot \dfrac{\rho_{\privacy_e, l}}{\rho_{\privacy_e, 1}} \cdot \left( 1 - p_{\privacy_e, 1} \right)^{n^+(l)} \cdot P_1(\Privacies^-(l))
\end{multline}
with
$$ n^+(l) = \sum_{\privacy \in \Privacies^+(l)} n_\privacy - 1$$

\begin{proof}
We use full-induction to prove the correctness of \autoref{eq:probability_calculation}.
For the base-case, we consider $\Privacies = \{ \privacy_e \}$. Based on \autoref{eq:base_equation_probability}, we observe that $P_1(\Privacies \setminus \privacy_e) = 1$ and $P_l(\Privacies \setminus \privacy_e) = 1$, as $\Privacies$ contains only $\privacy_e$. Additionally, $n^+(l) = n_{\privacy_e}-1$ for the same reason, which immediately leads to \autoref{eq:probability_calculation}.
For the induction step, we use $\Privacies^+_+(l) \subseteq \Privacies^+$ and $\Privacies^+_-(l) \subseteq \Privacies^+$ as auxiliary variables with $\privacy \in \Privacies^+_+(l) \oplus \Privacies^+_-(l), \forall \privacy \in \Privacies^+(l)$, for which the index states if they have already been included in the calculation.
Based on \autoref{eq:base_equation_probability} and \autoref{eq:probability_calculation}, we can derive \autoref{eq:full_induction_intermediate} associated $\privacy_e \in \Privacies^+_-(l)$ as intermediate state of the calculation.
Notice that $\privacy_e \in \Privacies^+_+$ by assumption.
Additionally, the privacy levels in $\Privacies^-$ are not considered on the left side of the equation, as $p_{\privacy, l} = 0, \forall \privacy \in \Privacies^-$.

\begin{multline}\label{eq:full_induction_intermediate}
    \impact_l \cdot \left( 1 - p_{\privacy_e, l} \right)^{n^+_+(l)} \cdot P_l(\Privacies^+_- \setminus \{\privacy_e\}) \\
    = \impact_1 \cdot \dfrac{\rho_{\privacy_e, l}}{\rho_{\privacy_e, 1}} \cdot \left( 1 - p_{\privacy_e, 1} \right)^{n^+_+(l)} \cdot P_1(\{\Privacies^-(l) \cup \Privacies^+_-(l)\})
\end{multline}
with
$$ n^+_+(l) = \sum_{\privacy \in \Privacies^+_+(l)} n_\privacy - 1$$

We aim to include a privacy level $\privacy_n$ into $\Privacies^+_+$.
Thus, we solve \autoref{eq:full_induction_intermediate} associated with $\privacy_n$ for $p_{\privacy_n, l}$ and insert it into \autoref{eq:full_induction_intermediate} associated with all other $\privacy_e \in \Privacies^+_-(l)\setminus\privacy_n$ to obtain \autoref{eq:full_induction_final}.

\begin{multline}\label{eq:full_induction_final}
    \impact_l \cdot \left( 1 - p_{\privacy_e, l} \right)^{n^+_+(l) + n_{\privacy_n}} \cdot P_l(\Privacies^+_- \setminus \{\privacy_e, \privacy_n\})
    = \impact_1 \cdot \\ \dfrac{\rho_{\privacy_e l}}{\rho_{\privacy_e 1}} \left( 1 - p_{\privacy_e, 1} \right)^{n^+_+(l) + n_{\privacy_n}} P_1(\{\Privacies^-(l) \cup \Privacies^+_-(l)\} \setminus \privacy_n)
\end{multline}

This equation is similar to our initial \autoref{eq:full_induction_intermediate} if we set $\Privacies^+_+ = \Privacies^+_+ \cup \privacy_n$ and $\Privacies^+_- = \Privacies^+_- \setminus \privacy_n$.
Additionally, it is evident that \autoref{eq:full_induction_final} is equal to \autoref{eq:probability_calculation} if $\Privacies^+_+ = \Privacies^+$ and $\Privacies^+_- = \emptyset$.

\end{proof}

\end{theorem}

\autoref{eq:probability_calculation} still contains $p_{\privacy_e, 1}$ as an auxiliary variable.
When replacing $p_{\privacy_e, 1}$ according to its definition in \autoref{eq:p1}, we can derive the remaining variables $p_{\privacy_e, i}, \forall i > 1$ only based on the other variables $p_{\privacy_e, i}, \forall i > 1$.
For that purpose, we introduce the variable $\dependency_l$ with $1 < l \leq n_\impact$ as defined in \autoref{eq:dependency}, which encapsulates the constant values and the dependency on other privacy levels $\privacy$ for readability.
Thus, we can transform \autoref{eq:probability_calculation} to \autoref{eq:probability_calculation_no_p1} by taking the $n^+(l)$-th root and replacing $p_{\privacy_e, 1}$.

\begin{algorithm}
\SetAlgoLined
\KwResult{$p_{\privacy, i}, \forall \privacy \in \Privacies, i \in \{1, \hdots, n_\impact\}$}
 $p_{\privacy, i} \leftarrow 0, \forall \privacy \in \Privacies, i \in \{1, \hdots, n_\impact\}$\;
 $c \leftarrow \infty$\;
 \For{$i \leftarrow 1$; $c > \epsilon$; $i \leftarrow (i\ \mathrm{mod}\ |\Privacies|) + 1$}{
    $\mathrm{temp}_j \leftarrow p_{i,j}, \forall j \in \{1, \hdots, n_\impact\}$\;
    $\mathrm{recal(p_{i,j})}, \forall j \in \{1, \hdots, n_\impact\}$\;
    $c \leftarrow \sum_{j = 1}^{n_\impact}|\mathrm{temp}_j-p_{i,j}|$\;
 }
 \Return $p_{\privacy, i}, \forall \privacy \in \Privacies, i \in \{1, \hdots, n_\impact\}$\;
 \caption{Determining the optimal strategy for all privacy-levels. \textbf{recal($\hdots$)} recalculates $p_{\privacy_e, i}$ based on the current values of $p_{\privacy, i}$. $\epsilon$ is the infinitesimal.}\label{alg:solve}
\end{algorithm}

\begin{equation}\label{eq:probability_calculation_no_p1}
    1 - p_{\privacy_e, l} = \left[1 - \left( \dfrac{A}{a_{\privacy_e, 1}} - \sum_{i=2}^{n_\impact} \dfrac{a_{\privacy_e, i} \cdot p_{\privacy_e, i}}{a_{\privacy_e, 1}} \right) \right] \cdot \dependency_l
\end{equation}
with
\begin{equation}\label{eq:dependency}
    \dependency_i = \sqrt[n^+(i)]{\left(\dfrac{\impact_1}{\impact_i}\right) \cdot \left(\dfrac{\rho_{\privacy_e, i}}{\rho_{\privacy_e, 1}}\right) \cdot \prod_{\privacy \in \Privacies^-(i)}(1 - p_{\privacy,1})^{n_\privacy}}
\end{equation}

The equation system described by \autoref{eq:probability_calculation_no_p1} for all $2 \leq l \leq n_\impact$ cannot be solved without considering the dependency on the other privacy levels encapsulated in $\dependency_l$.
However, this dependency is hard to resolve except for some special cases, as it removes the linearity from \autoref{eq:probability_calculation_no_p1}.
Thus, we assume that $\dependency_l$ is constant for the calculation of $p_{\privacy_e, l}, \forall\ l = \{2, \hdots, n_\impact\}$.
Thus, we can represent $p_{\privacy_e, j} \neq 0$ as $p_{\privacy_e, i} \neq 0$ by subtracting the representation of $p_{\privacy_e, i}$ from the representation of $p_{\privacy_e, j}$ according to \autoref{eq:probability_calculation_no_p1} and obtain \autoref{eq:representation_calculation}.

\begin{equation}\label{eq:representation_calculation}
    p_{\privacy_e, i} = \dependency_i \left( \dfrac{p_{\privacy_e, j} - 1}{\dependency_{j}} \right) + 1
\end{equation}

With this assumption, we can calculate every $p_{\privacy_l, l}$ with \autoref{eq:pl}, which can be derived from \autoref{eq:probability_calculation_no_p1} and the representation of any $p_{\privacy_e, i}$ as $p_{\privacy_e, j}$ from \autoref{eq:representation_calculation}.
Notice, that $\dependency_1 = 1$, as either $p_{\privacy,1} = 0$ (then $1 - p_{\privacy,1} = 1$ and disappears), or $p_{\privacy,1} \neq 0$ (then $\privacy \notin \Privacies^-(1)$).

\begin{equation}\label{eq:pl}
    p_{\privacy_e, l} = \dfrac{\left[A - \sum_{i = 1 | i \neq l \land \privacy_e \notin \Privacies^-(i)}^{n_\impact} a_{\privacy_e, i} \right]\dependency_l}{\sum_{i = 1 | i \neq l \land \privacy_e \notin \Privacies^-(i)}^{n_\impact} \left(a_{\privacy_e, i} \cdot \dependency_i\right)} + 1
\end{equation}

Based on \autoref{eq:pl}, we can determine the strategies for each privacy level using \autoref{alg:solve}. This algorithm ensures that the initial error (induced by setting all probabilities to 0) converges, i.e the initial error constantly reduces for each iteration of \autoref{alg:solve}.
This algorithm converges immediately if there is no inter-dependency between the privacy levels, \ie, if there is no other privacy level $\privacy_o\ | \ p_{\privacy_o, i} = 0$.
If there is an inter-dependency, it converges due to three factors: (i) In the calculation of $p_{\privacy, 1}$, all probabilities $p_{\privacy, i}$ with $i>1$ are utilized, thus, $p_{\privacy, 1}$ balances the error of the other probabilities.
(ii) $p_{\privacy, 1}$ influences $\dependency_i$ of all privacy levels in $\Privacies^-(i)$, but we can see that $\dependency_l$ in the nominator and $\dependency_i$ in the denominator partially cancel out the error of each other in \autoref{eq:pl}.
(iii) $\exists l,\privacy\ |\ n_\privacy < n^+(l)$, in which case the error in $\dependency_l$ gets reduced based on the errors of the other privacy levels.

%% file: snippets/game/solution.tex
In the previous section, we assumed that every probability under consideration is non-zero.
To calculate the overall optimal strategy, we consider every possible combination of zero and non-zero probabilities of every privacy level, \ie, we consider every possible combination of $\Privacies^+(l)$ and $\Privacies^-(l)$.
That is, the computational complexity of our approach is $\mathcal{O}(2^{|\Privacies| \cdot n_\impact})$, \ie,  is exponential with the number of privacy levels $|\Privacies|$ and the number of impact levels $n_\impact$.
This exponential growth is justified by the separate consideration of zero probabilities, which leads to $2$ tries per probability.
While an exponential growth is generally bad, we need to remember the limited size of $|\Privacies|$ and $n_\impact$.
As every single computation of probabilities is very fast, the total computation time of the probabilities remains comparably small (in our experiments, it stayed around $100ms$).
In the calculation, we set the probabilities of all $p_{\privacy, l} = 0\ |\ \privacy \in \Privacies^-(l)$ and only calculate the remaining probabilities with our approach proposed in the previous section.

%% file: snippets/analysis/intro.tex
In this section, we analyze the properties of the solution for our non-cooperative game and the implications for real-world applicability.
We start with a analysis of the properties of our approach.
Next, we perform a numerical analysis of our developed solution, with the goal to showcase the influence of location privacy on a heterogeneous network in a controlled environment.
The numerical analysis includes the bandwidth dependency of our approach, the performance depending on the neighborhood, and the consequences of different neighborhood of vehicles.

%% file: snippets/analysis/properties.tex
Our found solution has certain properties, on which we elaborate in the following.

\subsubsection{Optimality}
For each possible set of $\Privacies^+(j), \forall j \in \{1, \hdots, n_\impact\}$, the partial derivatives of the utility with respect to all probabilities are $0$, \ie, are either local optima or saddle points.
To proof that the found solutions are global optima, we need to ensure that there is no other optimum with a higher utility than the found solution.
For this purpose, we investigate on the second derivative of the utility function as shown in \autoref{eq:expected_utility_2nd_derivative}.

\begin{multline}\label{eq:expected_utility_2nd_derivative}
    \dfrac{\partial^2 \overline{u}}{\partial^2 p_{\privacy_{e},l}} = -\overline{\impact}_l \cdot \Psi_l - \overline{\impact}_1 \cdot \dfrac{a_l}{a_0} \cdot \left(-\dfrac{\rho_{\privacy_e, l}}{\rho_{\privacy_e, 1}}\right)^2 \cdot \Psi_1
\end{multline}
with
$$ \Psi_j = a_l \cdot n_{\privacy_{e}} \cdot (n_{\privacy_{e}} - 1) \cdot (1 - p_{\privacy_{e},l})^{n_{\privacy_e} - 2} \cdot P_j(\privacy \in \Privacies \setminus \{\privacy_e\}) $$

As $\Psi_i$, $\overline{\impact}_i$, and $a_i$ are non negative for all $i$, the second derivative of the utility with respect to any probability $p_{\privacy_e, l}$ is always smaller or equal to $0$.
Thus, the expected utility presented in \autoref{eq:expected_utility} is concave.
This guarantees that the found solution maximizes the utility, but is not necessarily unique, \ie, there might be other solutions with similar utility.

\subsubsection{Stability}
According to \autoref{theo:nash}, the found solution of our non-cooperative game is an unstable Nash equilibrium.
Thus, this equilibrium is only followed if every vehicle is aware that its neighbors follow this strategy.

\begin{theorem}\label{theo:nash}
The solution of our non-cooperative game shown in \autoref{eq:pl} is an unstable Nash equilibrium, \ie, no vehicle has an incentive to deviate from the found solution.

\begin{proof}
First, we show that \autoref{eq:pl} is a Nash equilibrium.
For that purpose, we modify \autoref{eq:expected_utility} such that the strategy $\vec{q}$ of the tagged vehicle is not necessarily similar to the strategy of its neighbors.
We refer to this modified utility with $\overline{u}^*(\privacy_e, q, \vec{p_0}, \hdots, \vec{p_{\left|\Privacies\right|}})$ as shown in \autoref{eq:expected_utility_ego}.
For better readability, we omit the parameters when clear from context.

\begin{equation}\label{eq:expected_utility_ego}
    \overline{u}^* = \sum_{i = 1}^{n_\impact} \overline{\impact}_{\privacy_e,i} \cdot a_{\privacy_e,i} \cdot \left[1 - P_i \cdot (1 - p_{\privacy_e})^{n_{\privacy_e, i} - 1} \cdot (1 - q_i)\right]
\end{equation}

By assumption, the tagged vehicle also needs to stick to its bandwidth requirements, thus, \autoref{eq:p1} holds.
As a lower utilization of bandwidth cannot increase the utility (as the bandwidth is not part of the utility function), we assume the bandwidth inequality to be an equality.
As the the modified expected utility is rather similar to the expected utility, we can derive:
$$ \dfrac{\partial \overline{u}^*}{\partial q_l} \cdot n_{\privacy_e} = \dfrac{\partial \overline{u}}{\partial p_{\privacy_e,l}}, \forall\ l \in \{1, \hdots, n_{\impact}\} $$
As the partial derivative of the expected utility is $0$ in our solution, it follows that the partial derivative of the modified expected utility is also $0$.
This is shown in \autoref{eq:nash}.
\begin{equation}\label{eq:nash}
\dfrac{\partial \overline{u}}{\partial p_{\privacy_e,l}} = 0 \Rightarrow \dfrac{\partial \overline{u}^*}{\partial q_l} = 0, \forall\ l \in \{1, \hdots, n_{\impact}\}    
\end{equation}
Additionally, the second derivative is always $0$, as the first derivative of modified expected utility $\overline{u}^*$ for any $q_l$ is independent of $q_l$.
Thus, any strategy that fully utilizes the available bandwidth for the tagged vehicle is optimal, as long as the derivative is valid, which assumes the same setup of zero and non-zero probabilities.
A different setup of zero and non-zero probabilities can never be better, as this setup is better for the overall solution, \ie, would have been selected.
Consequently, the found solution is an unstable Nash equilibrium.
\end{proof}
\end{theorem}
\begin{table}[t]
    \centering
    \begin{tabularx}{\linewidth}{l|X}
        Parameter & Value \\ \hhline{=|=}
%         Cluster Size & $5$ \\ \hline
        Bandwidth & $10\%$ of required \\ \hline
        Event impact (vector) & $(1, 10, 100, 1000)$ \\ \hline
        Event frequency (vector) & $(90\%, 9\%, 0.9\%, 0.1\%)$ \\ \hline
        Event distribution range (vector) & $(10km, 1km, 100km, 100km)$
    \end{tabularx}
    \caption{Default parameters of the numerical analysis and the evaluation.}
    \label{tab:parameters_numerical}
\end{table}
\subsubsection{Implicit Coordination of Responsibility}\label{subsec:implicit_coordination}
As described in \autoref{subsec:optimal}, our approach selects the optimal strategy for every privacy level by choosing a strategy composed of zero and non-zero probabilities.
Thus, our approach is capable of assigning tasks to certain privacy levels.
A single privacy level might be sufficient to ensure the reception of data of a certain privacy level, such that $|\Privacies^+(j)| = 1, \forall j \in \{1, \hdots, n_\impact\}$.
Thus, the other privacy levels do not aim to receive the messages of this impact level, but use their bandwidth for the reception of messages of different impact levels.
In general, privacy levels with a low imprecision aim to receive messages with high context-sensitivity, \ie, a small message dissemination distance $r_i$.
On the other hand, privacy-aware vehicles will aim to receive messages with a high message dissemination distance $r_i$, as the effects of location privacy is smaller for this type of messages as shown in \autoref{proof:bandwidth_decrease}.
This coordination of responsibility drastically improves the performance of our offloading approach, but also introduces additional challenges regarding the possibility of errors in this coordination.
We provide a deeper insight in these errors in \autoref{subsec:numerical}.

%% file: snippets/analysis/numerical.tex
In the following, we investigate two properties of our approach by performing a numerical analysis: (i) the compensation of negative impacts of location privacy, and (ii) the influence of wrong estimation of the neighborhood of vehicles in proximity.
The default parameters for the numerical evaluation can be found in \autoref{tab:parameters_numerical}.

\paragraph{Compensation for Location Privacy}
\autoref{fig:bandwidth_share_privacy} displays the ability of our proposed solution to compensate for location privacy of vehicles in the network.
It is evident, that the possibility of compensation depends on the degree of location privacy in the system, which is influenced by the radius of the imprecision area and the share of privacy-sensitive vehicles in the system.
\autoref{subfig:bandwidth_share_low_privacy} displays the influence of a comparably low level of location privacy, in which the radius of the imprecision area is only $100m$.
For this scenario, the influence of privacy is minimal, as the system also benefits from the possibility of coordination as described in \autoref{subsec:implicit_coordination}.
For a large area of imprecision like in \autoref{subfig:bandwidth_share_high_privacy}, the performance slightly decreases with an increasing share of privacy-sensitive vehicles until roughly $85\%$ of privacy-sensitive vehicles.
However, when we assume a system in which the majority of vehicles is privacy-sensitive, an increase in the share of privacy-agnostic vehicles provides a huge performance gain to the system.
This finding is very interesting, as it enables the compensation for location privacy of a subset of vehicles even for high privacy levels.

\begin{figure}
    \centering
    \subfloat[Low location privacy ($0.1km$).]{\includegraphics[width=0.45\linewidth]{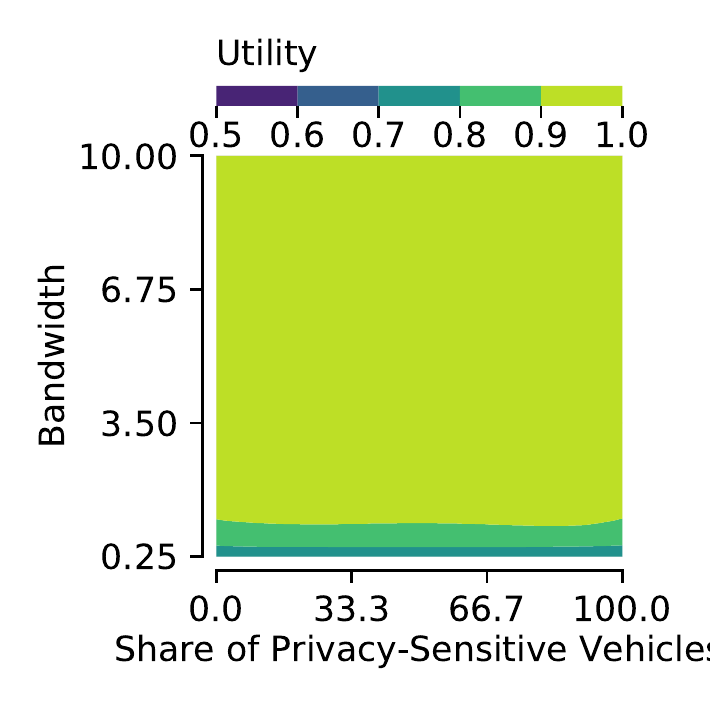}\label{subfig:bandwidth_share_low_privacy}} \quad
    \subfloat[High location privacy ($10km$).]{\includegraphics[width=0.45\linewidth]{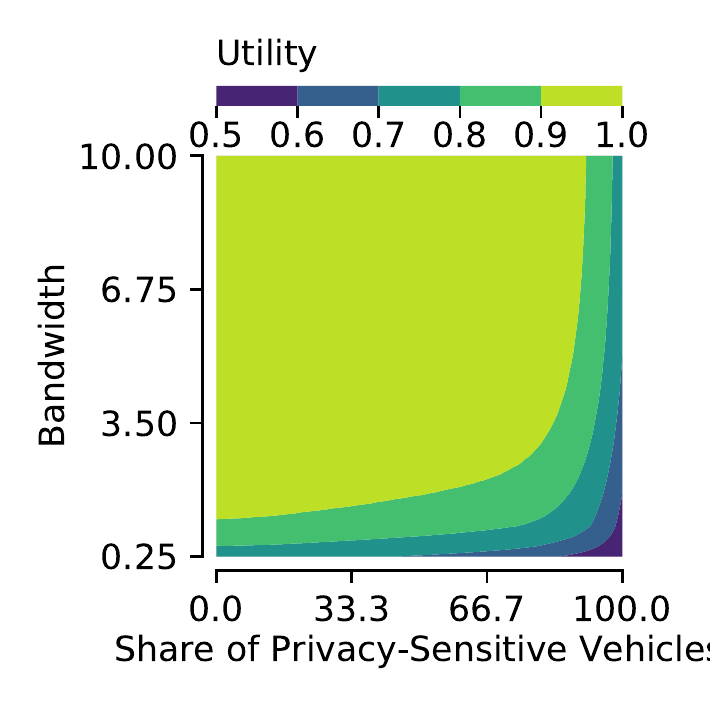}\label{subfig:bandwidth_share_high_privacy}}
    \caption{Utility depending on bandwidth and privacy.}
    \label{fig:bandwidth_share_privacy}
\end{figure}

\paragraph{Wrong Neighborhood Estimation}
\autoref{fig:diff_wrong_prediction} displays the changes in utility if the vehicles of a certain privacy-level have a different number of neighbors than the vehicle in another privacy level.
This leads to a wrong prediction of the strategy of the other privacy level, as the assumption of a similar neighborhood as presented in \autoref{sec:game} does not hold.
\autoref{subfig:overestimation} shows the influence of an overestimation of the number of neighbors, \ie, the neighborhood of a vehicle in proximity of the tagged vehicle is larger than the actual neighborhood of the tagged vehicle.
An overestimation of the number of neighbors seems to produce only minor issues, as the loss of utility is always below or equal $0.6\%$.
Compared to that, the underestimation severely impacts the performance of our approach, which leads to a loss of up to $50\%$ in certain situations.
Additionally, it can be observed that the loss in performance is only significant at the diagonal of the plot, \ie, if the number of privacy-sensitive and not privacy-sensitive vehicles is close.
This is justified as the privacy level with a higher influence (with more vehicles) has more possibility to cover certain impact levels, such that the other privacy levels cover the remaining impact levels.
However, if both privacy levels assume that they are the more influential, this approach fails, leading to impact levels that are not covered by any privacy level.
This leads to a severe drop in performance and needs to be prevented.
To achieve this, we need to restrict the implicit assignment of impact levels to certain privacy levels.
%For the explicit task assignment, there might be certain impact levels, which need to be always covered by a certain privacy level, which reduces the impact of the implicit assignment.
Thus, we perform a worst-case analysis, \ie, the vehicle checks what happens if its neighbors do not cooperate.
Depending on the trust in its neighbors, the expected utility with neighbors and the worst-case utility are weighted to receive the final score of a strategy.
%As both approaches, the explicit task assignment and the worst-case consideration, do not exclude each other, thus, we rely on both to improve the performance of our model under real-world conditions.

\begin{figure}
    \centering
    \subfloat[Overestimation.]{\includegraphics[width=0.45\linewidth]{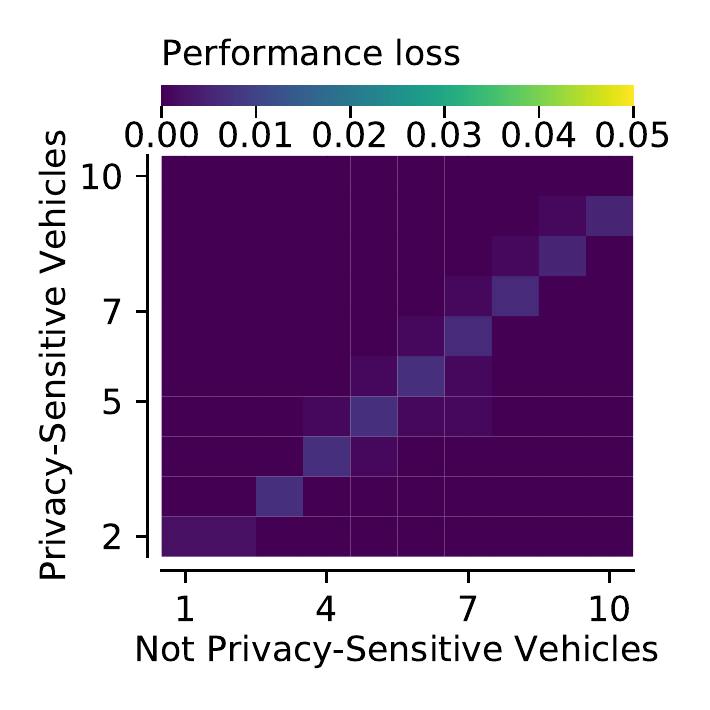}\label{subfig:overestimation}} \quad
    \subfloat[Underestimation.]{\includegraphics[width=0.45\linewidth]{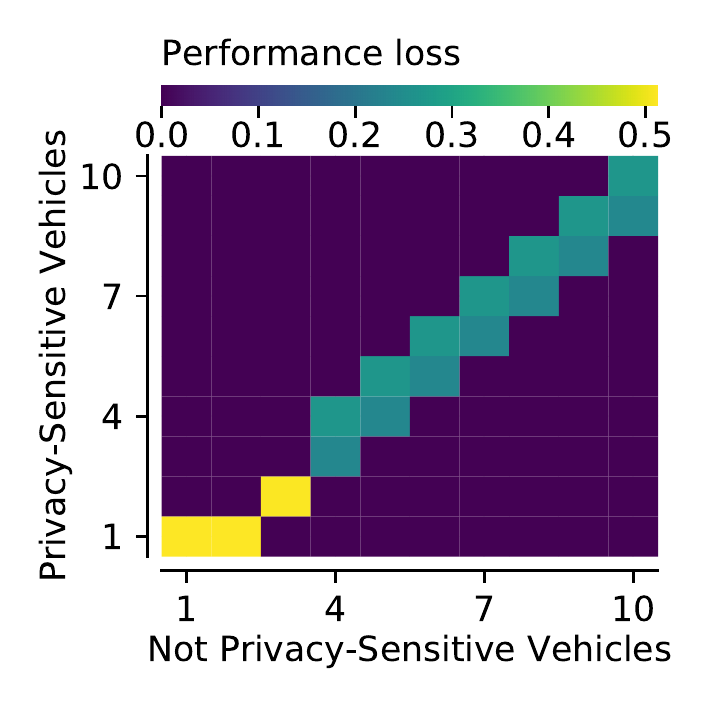}\label{subfig:underestimation}}
    \caption{Difference between expected utility and actual utility.}
    \label{fig:diff_wrong_prediction}
\end{figure}

%% file: chapter/4-evaluation.tex
In this section, we evaluate the performance of our approach in a realistic vehicular network under varying environmental conditions.
For this purpose, we utilize the vehicular extension of the Simonstrator framework \cite{MBSB19} in conjunction with SUMO \cite{SUMO2018} to simulate a vehicular network in Cologne~\cite{6117125}.
%In the simulation, we show the applicability of our approach in realistic vehicular network simulations.
We compare our approach with state-of-the-art methods for cooperative communication in large-scale vehicular networks and non-cooperative approaches.
In this large-scale vehicular network, messages are provided based on the current location of the vehicle (considering its privacy restrictions).
In our simulation, we generate messages randomly in an area of roughly $220 \times 220 km^2$, while the movement of vehicles and their networking is only simulated in an area of $2 \times 2 km^2$ to reduce the computational overhead.
As all events with a possible influence to the network are simulated, we accurately model the message load in a large-scale vehicular network.

\subsection{Reference Approaches}
In the following, we describe the evaluated approaches.

\subsubsection{Game-Theoretic Privacy-Sensitive Cooperation (GTP)}
This is our approach proposed in \autoref{sec:game}, which relies on implicit coordination between vehicles.

\subsubsection{No Cooperation (NC)}
The No-Cooperation (NC) approach does not consider cooperation between vehicles.
Thus, vehicles using the \textit{NC} approach receive similar messages as their neighbors, \ie, cannot share their messages.

\subsubsection{Clustering with perfect failure detection (GK)}
Clustering is a common strategy for the local distribution of messages in vehicular networks.
Instead of every vehicle communicating, only the so-called cluster-head is communicating directly with the server and provides received messages to vehicles in proximity via V2V communication.
This clustering approach can detect disconnects immediately and is used as an (unrealistic) upper bound for the performance of our approach.

\subsubsection{Clustering without perfect failure detection (CL)}
This clustering approach is similar to the previous approach, with the exception that the detection of disconnects is now imperfect.
Thus, the vehicles need to wait for a timeout until they detect a disconnect and reorganize the cluster.
This approach is more realistic than the previous \textit{GK} approach.

\subsection{Metrics}
We use two metrics to evaluate the performance of our approach: the \textit{achieved relative utility} and the \textit{used bandwidth}.
The \textit{achieved relative utility} measures the performance of the network, \ie, how much data is provided to a vehicle in the network.
This metric is between $0$ and $1$, where $1$ states that the vehicle has received everything that was sent and $0$ states that the vehicle has received nothing.
The \textit{used bandwidth} captures if the approaches stick to their average bandwidth limitations, \ie, if the side condition of the game is fulfilled.

\subsection{Plots}
We use box-plots and line-plots to visualize our results.
In the box-plots, the boxes show the differences between vehicles inside of one simulation run.
Next to each box, there is a line with a dot, visualizing the average value over all vehicles and simulation runs and the standard deviation of the average of all vehicles.
In line-plot, the line displays the mean value for the vehicles in one simulation run.

\subsection{Evaluation Results}
\begin{figure}
    \centering
    \subfloat[Achieved relative utility.]{\includegraphics[width=.45\linewidth]{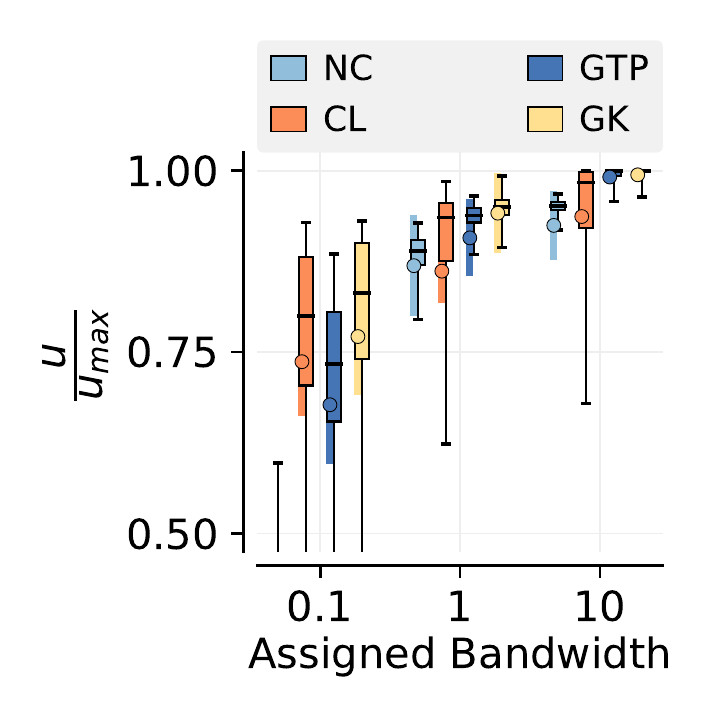}\label{subfig:eval_bandwidth_utility}} \quad
    \subfloat[High-impact.]{\includegraphics[width=.45\linewidth]{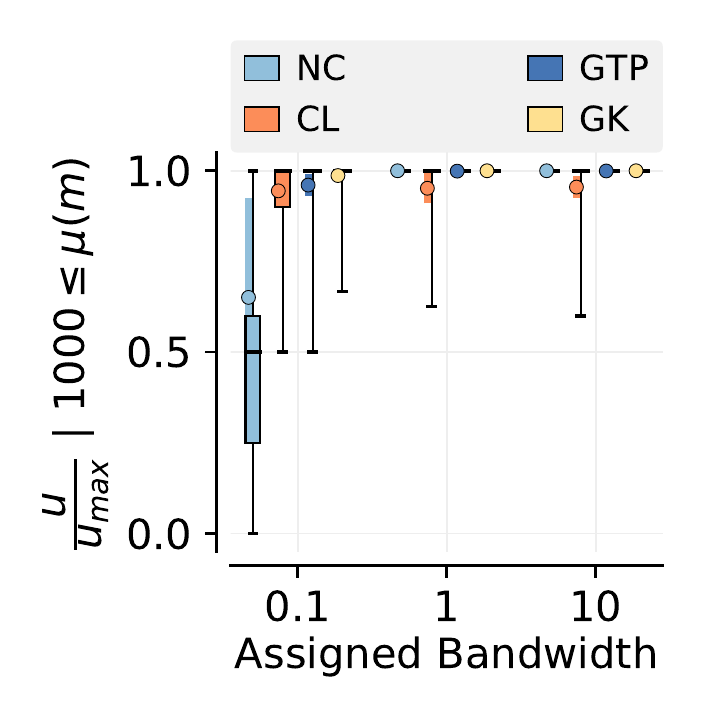}\label{subfig:eval_bandwidth_high_low}}
    \caption{Achieved relative bandwidth for different bandwidths.}
    \label{fig:eval_bandwidth}
\end{figure}
\autoref{fig:eval_bandwidth} depicts the performance of the approaches under different available bandwidths to each individual vehicle.
It is evident that the performance of all approaches increases with increasing bandwidth for all approaches as depicted by \autoref{subfig:eval_bandwidth_utility}.
For a full reception of all data available in the network via cellular, a bandwidth of roughly $100$ messages per second is required.
Even with a much smaller bandwidth of $10$ messages per second, all approaches can achieve reasonable utility levels by prioritizing high-impact messages.
It can be observed that our \textit{GTP} approach outperforms the \textit{CL} approach as well as the \textit{NC} approach and has much smaller confidence intervals compared to the \textit{CL} approach.
Thus, our approach is more resilient and adaptive to different network conditions.
Additionally, our approach is very close in performance to the \textit{GK} approach.
The same holds for a bandwidth of $1$, while our approach decreases in performance for a bandwidth of $0.1$.
For a bandwidth of $0.1$, our approach performs worse than the \textit{CL} approach, as the redundant transmission of high-impact messages and the missing explicit coordination between vehicles decrease the performance of our \textit{GTP} approach.
This is also confirmed by \autoref{subfig:eval_bandwidth_high_low}:
For the high-impact messages, our approach performs well for both a bandwidth of $1$ and $10$, but struggles to receives the high-impact messages for a bandwidth of $0.1$.
That is, a bandwidth of $0.1$ is not sufficient to receive the high-impact messages using only the available bandwidth of a single vehicle.
Thus, the performance of our approach decreases below the performance of the \textit{CL} approach, as the explicit coordination of vehicles in clustering approaches can handle low bandwidths well.
Additionally, all approaches stick to the available bandwidth on average, while the bandwidth is temporarily exceeded by a subset of vehicles.
This exceeding of bandwidth is justified by (i) the different number of available messages depending on the event location and (ii) the cooperative reception of messages by vehicles.

\autoref{fig:eval_privacy} displays the influence of privacy on our realistic vehicular network if the privacy-sensitive vehicles use an area of imprecision with radius $10km$.
\autoref{subfig:eval_privacy_all} shows the behavior of the relative utility for all of the approaches.
The \textit{NC} approach decreases the most, as the privacy-sensitive vehicles have no possibility to compensate for their context imprecision.
Additionally, our \textit{GTP} approach constantly outperforms the \textit{CL} approach and the \textit{NC} approach independent of the level of privacy.
Most interestingly, the performance decrease of our \textit{GTP} approach compared to the \textit{GK} approach is not constant, it is lowest around $50\%$ privacy.
This can be justified by implicit coordination between privacy levels as described in \autoref{subsec:implicit_coordination}.
This is also visible in \autoref{subfig:eval_privacy_low}, which displays the relative utility of messages with an impact between $10$ and $100$.
While the \textit{NC} approach is not able to receive this messages at all, the utility of the other approaches decreases constantly.
However, for our \textit{GTP} approach, the utility remains constant for a very long duration, which leads to a comparably constant overall utility even for high privacy levels.

\begin{figure}
    \centering
    \subfloat[Overall.]{\includegraphics[width=.45\linewidth]{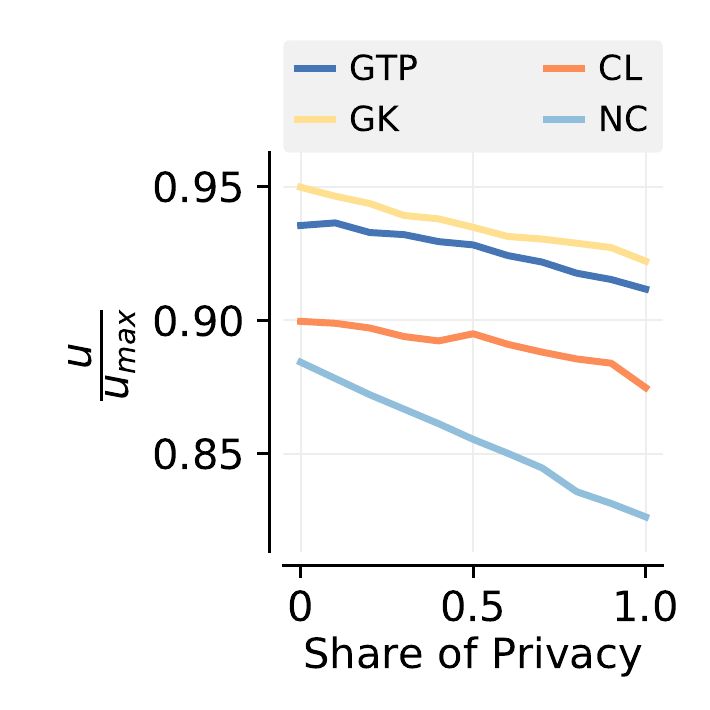}\label{subfig:eval_privacy_all}} \quad
    \subfloat[Low-impact.]{\includegraphics[width=.45\linewidth]{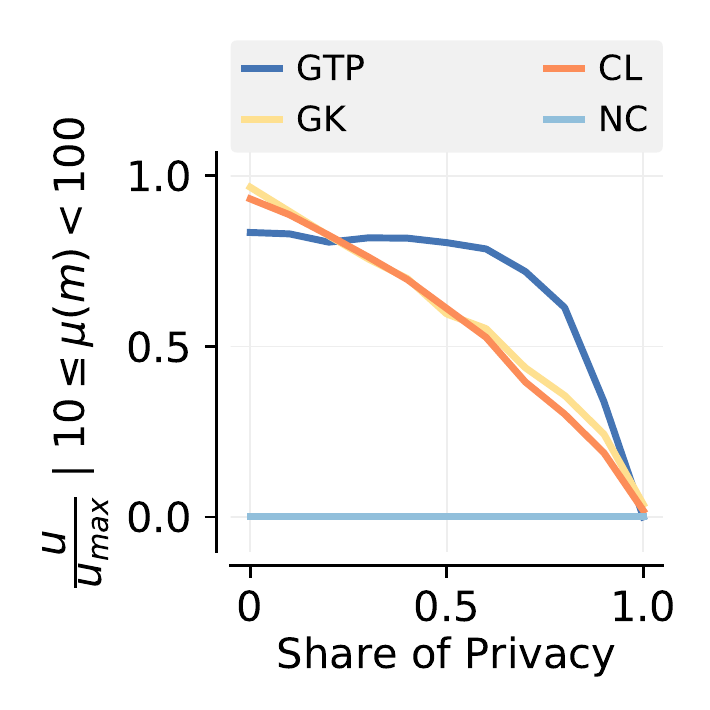}\label{subfig:eval_privacy_low}}
    \caption{Achieved relative utility for mixed environments.}
    \label{fig:eval_privacy}
\end{figure}

%% file: chapter/5-conclusion.tex
In this paper we introduce privacy considerations in the management of \gls{fcd} and have shown its impact on location-based services, since some data are not forwarded to a vehicle due to privacy considerations and the implemented location obfuscation. In order to alleviate this problem, we have introduced cooperation among vehicles so as to forward relevant data to their neighboring vehicles, enhancing in principle the data received by a vehicle only directly from the remote server. In this work, an ad-hoc, direct V2V cooperation paradigm is employed instead of a cluster-based one, also showing the high performance deterioration of the latter in a real vehicular networking environment.
A major contribution of this work is the development and study of a non-cooperative game determining the strategies (in terms of probabilities that a vehicle is forwarded by the server data of a given impact index) that vehicles should follow, so that a properly defined utility is maximized; this is shown to lead to a diversification of the data received directly from the server by neighboring vehicles and increases the effectiveness of V2V cooperation. 

Our numerical analysis show that the influence of privacy can be compensated as long as the share of privacy-sensitive vehicle is below $85\%$.
Above that threshold, an increase in the share of privacy-sensitive vehicles drastically decreases the performance of the system.
In the evaluation, we analyzed the performance of our approach in a realistic vehicular network.
Our results show the drastic performance increase compared to non-cooperative approaches and the improvements over cluster-based approaches.
Additionally, our approach performs almost similarly to a perfect clustering approach, which utilizes bandwidth optimally and detects disconnects immediately, but is not realizable in reality.
When we analyze the performance of our approach for different privacy levels, we see that the performance remains constant for a long time, which confirms the results from our numerical analysis.